\documentclass[runningheads]{llncs}

\usepackage{ltl}
\usepackage{mathtools}
\usepackage{xspace}
\usepackage{algorithm}
\usepackage[noend]{algpseudocode}
\usepackage{enumitem}
\usepackage{amssymb}
\usepackage{amsmath}
\usepackage{graphicx}
\usepackage{amsfonts} 
\usepackage{stmaryrd} 
\usepackage{xcolor}
\usepackage{booktabs}
\usepackage{stmaryrd}
\usepackage{thm-restate}
\usepackage{upgreek}

\definecolor{linkC}{HTML}{710580}
\definecolor{labels}{HTML}{cc0000}
\usepackage{hyperref}      
\hypersetup{
	colorlinks,
	citecolor=linkC,
	filecolor=red,
	linkcolor=linkC,
	urlcolor=blue
}             

\usepackage{tablefootnote}
\usepackage{makecell}
\usepackage{todonotes}
\usepackage{comment}
\usepackage{subcaption}
\usepackage{xcolor}
\usepackage{listings}
\usepackage{relsize}

\usepackage{thm-restate}

\usepackage{bussproofs}

\usepackage{booktabs}

\usepackage{pifont}

\usepackage{wrapfig}

\usepackage{multirow}

\usepackage{hyperref}     
\usepackage{cleveref}
\crefname{theorem}{Thm.}{Theorems}
\Crefname{theorem}{Thm.}{Theorems}
\crefname{example}{Ex.}{Examples}
\Crefname{example}{Ex.}{Examples}
\crefname{equation}{Eq.}{Eqs.}
\Crefname{equation}{Eq.}{Eqs.}
\crefname{figure}{Fig.}{Figs.}
\Crefname{figure}{Fig.}{Figs.}
\crefname{algorithm}{Algo.}{Algos.}
\Crefname{algorithm}{Algo.}{Algos.}
\crefname{definition}{Def.}{Defs.}
\Crefname{definition}{Def.}{Defs.}
\usepackage{orcidlink}

\newcommand\donotshow[1]{}

\usepackage{algorithm}
\allowdisplaybreaks
\makeatletter
\newcommand{\superimpose}[2]{{%
		\ooalign{%
			\hfil$\m@th#1\@firstoftwo#2$\hfil\cr
			\hfil$\m@th#1\@secondoftwo#2$\hfil\cr
		}%
}}
\makeatother

\newcommand{\nat}{\mathbb{N}}

\newcommand{\true}[0]{\mathtt{true}}
\newcommand{\false}[0]{\mathtt{false}}

\makeatletter
\DeclareFontFamily{OMX}{MnSymbolE}{}
\DeclareSymbolFont{MnLargeSymbols}{OMX}{MnSymbolE}{m}{n}
\SetSymbolFont{MnLargeSymbols}{bold}{OMX}{MnSymbolE}{b}{n}
\DeclareFontShape{OMX}{MnSymbolE}{m}{n}{
    <-6>  MnSymbolE5
   <6-7>  MnSymbolE6
   <7-8>  MnSymbolE7
   <8-9>  MnSymbolE8
   <9-10> MnSymbolE9
  <10-12> MnSymbolE10
  <12->   MnSymbolE12
}{}
\DeclareFontShape{OMX}{MnSymbolE}{b}{n}{
    <-6>  MnSymbolE-Bold5
   <6-7>  MnSymbolE-Bold6
   <7-8>  MnSymbolE-Bold7
   <8-9>  MnSymbolE-Bold8
   <9-10> MnSymbolE-Bold9
  <10-12> MnSymbolE-Bold10
  <12->   MnSymbolE-Bold12
}{}

\let\llangle\@undefined
\let\rrangle\@undefined
\DeclareMathDelimiter{\llangle}{\mathopen}%
                     {MnLargeSymbols}{'164}{MnLargeSymbols}{'164}
\DeclareMathDelimiter{\rrangle}{\mathclose}%
                     {MnLargeSymbols}{'171}{MnLargeSymbols}{'171}
\makeatother

\definecolor{dkcyan}{rgb}{0.1, 0.3, 0.3}
\definecolor{dkgreen}{rgb}{0,0.3,0}
\definecolor{olive}{rgb}{0.5, 0.5, 0.0}
\definecolor{dkblue}{rgb}{0,0.1,0.5}

\definecolor{col:ln}{rgb}  {0.1, 0.1, 0.7}
\definecolor{col:str}{rgb} {0.8, 0.0, 0.0}
\definecolor{col:db}{rgb}  {0.9, 0.5, 0.0}
\definecolor{col:ours}{rgb}{0.0, 0.7, 0.0}

\definecolor{lightgreen}{RGB}{170, 255, 220}
\definecolor{darkbrown}{RGB}{121,37,0}

\colorlet{listing-comment}{gray}
\colorlet{operator-color}{darkbrown}
\colorlet{comment-color}{black!50}

\lstdefinelanguage{custom-lang}{
	keywords={let, in, where, match, with, when, if, then, else, for, repeat, return, to, do, from},
	keywordstyle=[1]\color{dkblue},
	morekeywords=[2]{verify, systemToNBA, LTLtoNBA, eProduct, uProduct},
	keywordstyle=[2]\color{dkgreen},
    morekeywords=[3]{underApprox, overApprox},
	keywordstyle=[3]\color{darkbrown},
	comment=[l][\color{comment-color}]{//},
	literate=%
	{=}{{{\color{operator-color}=}}}1
	{|}{{{\color{dkblue}|}}}1
	{:}{{{\color{dkblue}:}}}1
	{:=}{{{\color{dkblue}:=}}}1
    {@}{ }1
}

\lstdefinestyle{default}{
	escapeinside={(*}{*)},
	basicstyle=\ttfamily\fontsize{9.3}{10.3}\selectfont,
	columns=fullflexible,
	commentstyle=\sffamily\color{black!50!white},
	framexleftmargin=1em,
	framexrightmargin=1ex,
	keepspaces=true,
	keywordstyle=\color{dkblue},
	mathescape,
	numbers=left,
	numberblanklines=false,
	numbersep=1.25em,
	numberstyle=\relscale{0.8}\color{gray}\ttfamily,
	showstringspaces=true,
	stepnumber=1,
	xleftmargin=2em,
}

\lstnewenvironment{code}[1][]
{\small
	\lstset{
		style=default, 
		language=custom-lang,
		#1
	}
}
{}

\lstdefinelanguage{example-lang}{
	keywords={while,do},
	keywordstyle=[1]\bfseries,
	comment=[l][\color{comment-color}]{//},
	literate=%
	{<-}{{{\color{dkblue}$\leftarrow$}}}1
	{@}{ }1
}

\lstdefinestyle{example-style}{
	escapeinside={(*}{*)},
	basicstyle=\ttfamily\fontsize{8.4}{9.7}\selectfont,
	columns=fullflexible,
	commentstyle=\sffamily\color{black!50!white},
	framexleftmargin=0em,
	framexrightmargin=0ex,
	keepspaces=true,
	keywordstyle=\color{dkblue},
	mathescape,
	numbers=left,
	numberblanklines=false,
	showstringspaces=true,
	stepnumber=1,
	xleftmargin=0em,
	numbers=none
}

\lstnewenvironment{exampleCode}[1][]
{\small
	\lstset{
		style=example-style, 
		language=example-lang,
		#1
	}
}
{}

\usepackage{scalefnt}

\newcounter{claim} 
\renewcommand{\theclaim}{\arabic{claim}}

\crefname{claim}{claim}{claims}

\renewenvironment{proof}[1][Proof]{%
    \par\noindent\textit{#1.}\hspace{0.5em}\ignorespaces%
}{%
    \hfill$\qed$%
    \par\vspace{0.5em}%
}

\newcommand{\abs}[1]{\left\lvert#1\right\rvert}
\newcommand{\last}{\mathit{last}}
\newcommand{\lang}{\mathcal{L}}
\newcommand{\dom}{\mathit{dom}}

\newcommand{\ap}{\mathtt{AP}}
\newcommand{\trace}{\gamma}

\newcommand{\traces}{\mathit{Traces}}

\newcommand{\globally}{\LTLglobally}
\newcommand{\eventually}{\LTLfinally}
\newcommand{\until}{\ U\ }
\newcommand{\weakuntil}{\ W\ }
\newcommand{\nextt}{\LTLnext}

\newcommand{\architecture}{\mathtt{Arc}}
\newcommand{\proc}{\text{Proc}}
\newcommand{\plant}{\texttt{p}}
\newcommand{\plantall}{\mathbb{P}}
\newcommand{\contr}{\texttt{c}}
\newcommand{\contrall}{\mathbb{C}}
\newcommand{\env}{\texttt{e}}
\newcommand{\inputs}[1]{{I_{#1}}}
\newcommand{\outputs}[1]{{O_{#1}}}

\newcommand{\run}{\rho}
\newcommand{\runt}{\mathit{run}}
\newcommand{\runs}{\mathit{Runs}}
\newcommand{\tree}{\mathcal{T}}
\newcommand{\runtree}{\upeta}

\newcommand{\aut}{\mathcal{A}}
\newcommand{\acc}{\Omega}
\newcommand{\safety}{\mathit{safe}}

\newcommand{\reach}{\texttt{reach}}

\newcommand{\strat}{\mathcal{M}}
\newcommand{\moore}{\strat}

\newcommand{\stratoutput}{\mathit{out}}

\newcommand{\labelts}{o}
\newcommand{\univcontr}{\mathcal{U}}
\newcommand{\gencontr}{\univcontr}
\newcommand{\hyperstrat}{\univcontr}

\newcommand{\ltlToAut}{\textsf{LTLtoDSA}}

\newcommand{\prophecy}{\theta}
\newcommand{\safeProphecy}{\mathbb{S}}
\newcommand{\safeProphecyF}[2]{\safeProphecy_{#1}^{#2}}
\newcommand{\prophecyall}{\mathbb{F}}

\newcommand{\bracket}[1]{\langle #1\rangle}

\newcommand{\prototype}{\textsc{unicon}\xspace}

\lstdefinelanguage{custom-lang}{
	keywords={let, in, match, with, when, if, then, else, elif, for, to, do, rec, return, new, not, and, while},
	keywordstyle=[1]\color{dkblue}\bfseries,
	morekeywords=[2]{append, Set, Dict, Queue, pop, push, add, contains},
	keywordstyle=[2]\sffamily,
	morekeywords=[3]{generalController, prophecy, existentialProjection, composition, isMember, stepComposition, modifiedPlant, h},
	keywordstyle=[3]\color{dkcyan}\ttfamily,
	comment=[l][\color{comment-color}]{//},
	literate=%
	{=}{{{\color{operator-color}=}}}1
	{<-}{{{\color{operator-color}$\leftarrow$}}}1
	{|}{{{\color{dkblue}$\mid$}}}1
	{:}{{{\color{dkblue}:}}}1
	{:=}{{{\color{dkblue}:=}}}1
	{@}{ }1
}

\lstdefinestyle{default}{
	escapeinside={(*}{*)},
	basicstyle=,
	columns=fullflexible,
	commentstyle=\sffamily\color{black!50!white},
	framexleftmargin=1em,
	framexrightmargin=1ex,
	keepspaces=true,
	keywordstyle=\color{dkblue},
	mathescape,
	numbers=left,
	numberblanklines=false,
	numbersep=0.5em,
	numberstyle=\relscale{0.75}\color{gray}\ttfamily,
	showstringspaces=true,
	stepnumber=1,
	xleftmargin=1.2em,
}

\lstnewenvironment{mycode}[1][]
{\small
	\lstset{
		style=default, 
		language=custom-lang,
		#1
	}
}
{}

\usetikzlibrary{arrows}
\usepackage{csquotes}

\begin{document}

\title{Synthesis of Universal Safety Controllers}

 \author{Bernd Finkbeiner\inst{1}\orcidlink{0000-0002-4280-8441} \and Niklas Metzger\inst{1}\orcidlink{0000-0003-3184-6335} \and Satya Prakash Nayak\inst{2}\orcidlink{0000-0002-4407-8681} \and \\ Anne-Kathrin Schmuck\inst{2}\orcidlink{0000-0003-2801-639X}}

 \authorrunning{B. Finkbeiner, N. Metzger, S. P. Nayak, and A. Schmuck}
 %
 \institute{CISPA Helmholtz Center for Information Security, Saarbrücken, Germany\\
 \email{\{finkbeiner,niklas.metzger\}@cispa.de} \and
Max Planck Institute for Software Systems, Kaiserslautern, Germany\\
 \email{\{sanayak,akschmuck\}@mpi-sws.org}\\
 \url{}}

\maketitle

\begin{abstract}
The goal of \emph{logical controller synthesis} is to automatically compute a control strategy that regulates the discrete, event-driven behavior of a given plant s.t.\ a temporal logic specification holds over all remaining traces. Standard approaches to this problem construct a two-player game by composing a given complete plant model and the logical specification and applying standard algorithmic techniques to extract a control strategy. However, due to the often enormous state space of a complete plant model, this process can become computationally infeasible. In this paper, we introduce a novel synthesis approach that constructs a \emph{universal controller} derived solely from the game obtained by the standard translation of the logical specification. The universal controller's moves are annotated with \emph{prophecies} -- predictions about the plant's behavior that ensure the move is safe. By evaluating these prophecies, the universal controller can be adapted to any plant over which the synthesis problem is realizable. This approach offers several key benefits, including enhanced scalability with respect to the plant's size, adaptability to changes in the plant, and improved explainability of the resulting control strategy. We also present encouraging experimental results obtained with our prototype tool, \prototype{}.
\end{abstract}

\section{Introduction}\label{sec:introduction}

Reactive synthesis, in the tradition of Church's Problem~\cite{Church1964}, and
supervisory control, in the tradition of Ramadge and Wonham~\cite{10.1137/0325013}, both aim
at the automatic construction of systems that are guaranteed to
satisfy their formal specifications\footnote{See~\cite{EhlersJdeds,SchmuckJdeds20} on an in-depth discussion of their resulting algorithmic differences.}.  Reactive synthesis is based on a logical
specification, typically given as a formula of a temporal logic;
supervisory control is based on an open system model, known as
the \emph{plant}. In practical applications, we often encounter a
combination of both views, where, for example, the behavioral objectives of a robot
are specified as a temporal formula, while its
surroundings, such as walls, humans, and other robots, as well as the physics of their interaction, are described
by a plant model.
This has lead to approaches which automatically abstract physical behavioral constrains into a plant automaton~\cite{tabuada2009verification,belta2017formal,Review24_FMandControl_YinGaoYu} or model the plant directly as an event-based automaton~\cite{kress2018synthesisForRobots_Review,wonham2018supervisory}, which is then composed with the finite automaton derived from the logical specification. The resulting $\omega$-regular game can then be solved by standard algorithmic techniques from reactive synthesis~\cite{AbadiLW89,JobstmannB06,JobstmannGWB07}. 
This approach is, however, severely limited in its scalability. The necessity to include \emph{all possible} plant behaviors into the plant model to achieve a correct-by-design solution leads to a huge state space of the plant automaton, which makes solving the resulting game computationally impractical.
The situation is even worse if the plant model is not fully known or changes over
time, for example as the robot receives new mapping data, which currently requires a complete re-synthesis. 

An encouraging observation is, however,  that, in practice, it is often not
necessary to analyze the interaction of the controller and the plant
in great detail. In supervisory control, it has been observed that, if
certain conditions are met, the plant already guarantees the desired
behavior entirely on its own, independently of the
controller~\cite{8843071}. Conversely, benchmarks in reactive synthesis, where there
is no explicit plant model, often capture the relevant properties of
the plant with just a small set of assumption formulas~\cite{DBLP:journals/corr/abs-2206-00251}.

Based on this observation, this paper presents a novel 
synthesis approach that decouples the synthesis of a \emph{universal controller}, which we construct purely based on the temporal formula, from the
\emph{adaptation} of the universal controller to the specific plant model. The two phases are connected by an annotation of the
output decisions in the universal controller with a condition about the plant behavior, which we call the
\emph{prophecy} of the output: the controller only chooses the output 
if the plant satisfies the prophecy.  The main advantage of the new approach is the improved scalability due to the fact that the universal controller is
constructed \emph{independently of the plant}. 
We illustrate the idea
with a toy example, which will also serve as the running example of
the paper.

\begin{example}\label{example:running}
Our setting consists of the controller $\contr$, the plant $\plant$, and the external environment $\env$. We assume that each of these three processes has a boolean output proposition: $o_\contr$, $o_\plant$, and $o_\env$, respectively.
At each time step, the processes chose the values for their respective output propositions based on the history of all previous outputs. The temporal specification is given as the LTL formula
\begin{equation}\label{equ:spec:exp}
 \varphi \coloneqq (o_\contr \leftrightarrow o_\plant) \weakuntil (\nextt \neg o_\env).
\end{equation}
The formula requires that the controller matches the output of the plant as long as the environment sets its output to true.
As the controller cannot observe the plant output before setting its own output, it is impossible to build a correct controller independently of the plant:
unless the environment sets its output to false, the controller must know the output produced by the plant one time step in advance.
We express this information with the prophecies $\bracket{o_\plant}$, indicating that the plant will output true, and $\bracket{\neg o_\plant}$, indicating that the plant will output false.
The universal controller then reacts to $\bracket{o_\plant}$ by setting $o_\contr$ to true and to $\bracket{\neg o_\plant}$ by setting $o_\contr$ to false. 
\end{example}

\subsubsection*{Contribution.}
In the remainder of the paper, we formally introduce the concepts of universal
controllers and prophecies.
We provide an algorithm to synthesize 
universal controllers from safety specifications, with a finite representation of the  
prophecies as tree automata. Furthermore, we show that 
universal controllers are also \emph{most permissive}, capturing all
controller strategies that are correct for any realizable plant
model. We present an exploration algorithm to obtain a
controller from the universal controller for a specific plant
model. We show that the algorithm is sound and complete for safety
specifications. Finally, we provide experimental results based on our
prototype tool, \prototype{}~\cite{tool}, showing that our approach scales better
than the standard method on synthesis problems with large plant models.

 \subsubsection*{Related Work.}

One key feature of the universal controller is its \emph{permissiveness}: it captures all control strategies that are correct for any realizable plant.
The concept of permissiveness is well-established in the field of supervisory control, where it is \emph{the} key feature of the computed supervisor~\cite{cassandras2021introduction}. 
However, here permissiveness is understood w.r.t.\ enabling all behaviors that a \emph{single plant} fulfills. In the same spirit, permissive solutions have also been considered in the field of reactive synthesis~\cite{bernet2002:permissiveStrategies,bouyer2011:measuringpermissiveness,FremontSeshi_Improvisation}, however mostly to ensure variability of resulting implementations. Recently, permissiveness for control applications of reactive synthesis has been considered~\cite{Klein2015:mostGeneralController,AnandNS23}, but also w.r.t.\ one fixed plant model.

Another line of work focuses on synthesizing control strategies that are correct for a particular set of plant models. 
For example, \emph{dominant strategies}~\cite{DammF14,FinkbeinerP22} are correct for all realizable plants, whereas \emph{admissible strategies}~\cite{Berwanger07,BassetRS17} are correct for a maximal set of realizable plants. Other works~\cite{AnandMNS23,BrenguierRS17,DBLP:conf/cav/FinkbeinerMM22,DBLP:conf/cav/FinkbeinerMM24} restrict plant models with reasonable assumptions and synthesize a control strategy that is correct for this restricted set of plants.
Nevertheless, all these works are restricted by computing \emph{a single control strategy} that must be correct for the entire set of considered plants, whereas our approach computes a universal controller that can be adapted to \emph{different control strategies} depending on the plant model.

Finally, the use of prophecy variables is a common proof technique to make information about future events accessible~\cite{AbadiL91}. A recent work by Beutner and Finkbeiner~\cite{BeutnerF22} uses such prophecies to provide a complete algorithm for hyperproperty verification. Inspired by this work, we use prophecies to guess the future behavior of the plant in order to synthesize a universal controller.

\section{Preliminaries}
For every set $X$, we write $X^*$ and $X^\omega$ to denote the sets of finite and respectively infinite, sequences of elements of $X$, and let $X^\infty = X^* \cup  X^\omega$. 
For $\run\in X^\infty$, we denote $\abs{\run}\in\nat\cup \{\infty\}$ the length of $\run$, and define $\dom(\run)\coloneqq \{0,1,\ldots,\abs{\run}-1\}$.
For $\run = x_0x_1\cdots\in X^\infty$ and $i,j\in\dom(\run)$ with $i\leq j$, we define $\rho[i]\coloneqq x_i$ and $\rho[i,j]\coloneqq x_i\cdots x_j$. 
Moreover, $\last(\run)$ denotes the last element of $\run$.
For $\run\in X^\infty$ and some set $Y$, we write $\run\downarrow_Y$ to denote the sequence $\run'$ with $\run'[i] = \run[i]\cap Y$.

\medskip
\noindent\textit{System Architectures.}
In this paper, we focus on systems with three processes, $\proc = \{\contr,\plant,\env\}$, i.e., a controller, a plant, and an environment.
For the rest of the paper, we fix a set of atomic propositions $\ap = \outputs{\env}\uplus\outputs{\contr}\uplus\outputs{\plant}$ partitioned into the output propositions of the processes.
This defines an architecture $\architecture = (\inputs{\env},\outputs{\env},\inputs{\contr},\outputs{\contr},\inputs{\plant},\outputs{\plant})$ over $\ap$ where the input propositions of a process $i\in\proc$ are the union of the output propositions of all other processes, i.e., $\inputs{i} = \bigcup_{j\neq i}\outputs{j}$.

\medskip
\noindent\textit{Traces and Specifications.}
We fix alphabet $\Sigma = 2^{\ap}$ to the power set of the atomic propositions $\ap$.
A \emph{trace} over $\Sigma$ is a sequence $\trace\in \Sigma^\omega$.
A \emph{specification} $\varphi$ specifies some restrictions on the traces. 
We write $\trace \vDash\varphi$ to denote that the trace $\trace$ satisfies the specification $\varphi$.
The language $\lang(\varphi)$ of a specification $\varphi$ represents the set of traces that satisfy~$\varphi$.

\medskip
\noindent\textit{LTL.}
Linear-time temporal logic~(LTL)~\cite{Pnueli77} is a specification language for linear-time properties. 
LTL specifications over atomic propositions $\ap$ are given by
$ \varphi, \psi ::= \alpha\in\ap ~ | ~ \neg \varphi ~ | ~ \varphi \lor \psi ~ | ~ \varphi \land \psi ~ | ~  \nextt \varphi ~ | ~ \varphi \until \psi ~ | ~ \varphi \weakuntil \psi ~ | ~ \eventually \varphi ~ | ~ \globally \varphi$.
A trace is evaluated against an LTL specification $\varphi$ in the ususal way. We refer to~\cite[Chapter~5.1.2]{baier2008principles} for the detailed semantics. 

\medskip
\noindent\textit{Process Strategies.}
We model a strategy for process ${i}$ as a Moore machine~$\moore = (S, s_0,\tau,\labelts)$ over inputs $\inputs{i}$ and outputs $\outputs{i}$, consisting of a finite set of states $S$, an initial state $s_0 \in S$, a transition function $\tau: S \times 2^\inputs{i} \rightarrow S$ over inputs, and an output labeling function $\labelts: S \rightarrow 2^\outputs{i}$.
For a finite input sequence $\trace = \alpha_0\alpha_1\cdots\alpha_{k-1}\in (2^{\inputs{i}})^*\!$, $\moore$ produces a finite path $s_0s_1\cdots s_k$ and an output sequence $\labelts(s_0)\labelts(s_1)\cdots\labelts(s_k)\in (2^{\outputs{i}})^*$ such that $\tau(s_j, \alpha_j) = s_{j+1}$.
We write $\stratoutput(\strat,\trace)$ to denote the last output produced, i.e., $\labelts(s_k)$.
Similarly, for an infinite input sequence $\trace\in (2^{\inputs{i}})^\omega\!$, $\strat$ produces an infinite path $s_0s_1\cdots$ and an infinite output sequence $\labelts(s_0)\labelts(s_1)\cdots\in (2^{\outputs{i}})^\omega$.
We write $\traces(\strat)$ to denote the set of all traces $\trace\in\Sigma^\omega$ such that for input sequence $\trace\downarrow_{\inputs{i}}$, $\strat$ produces the output sequence $\trace\downarrow_{\outputs{i}}$.
We say that $\strat$ satisfies a specification~$\varphi$, denoted $\strat \models \varphi$, if $\trace \models \varphi$ holds for all traces $\trace\in\traces(\strat)$.

For ease of readability, we refer to the plant strategies as \emph{plants} and the controller strategies as \emph{controllers}.
We collect all plants in the set $\plantall$ and all controllers in the set $\contrall$.
The parallel composition $\moore_{i} \parallel \moore_j$ of two strategies $\moore_{i} = (S_i,s^i_0,\tau_i,o_i)$, $\moore_j = (S_j,s^j_0,\tau_j,o_j)$ of processes $i, j\in\proc$ is a strategy, i.e., a Moore machine, $(S,s_0,\tau,o)$ over inputs $(\inputs{i} \cup \inputs{j})\setminus(\outputs{i}\cup \outputs{j})$ and outputs $\outputs{i} \cup \outputs{j}$ with $S = S_i \times S_j$, $s_0=(s^i_0,s^j_0)$, 
$\tau((s,s'),\sigma) = (\tau_i(s,(\sigma \cup \labelts_j(s')) \cap \inputs{i}),\tau_j(s',(\sigma \cup \labelts_i(s))\cap\inputs{j}))$, and 
$\labelts((s,s')) = \labelts_i(s) \cup \labelts_j(s')$.

We also represent a strategy $\moore$ over inputs $I$ and outputs $O$ as a $2^{O}$-labeled $2^{I}$-tree which is a function $\tree : (2^{I})^* \rightarrow 2^{O}$.
The tree maps input sequences to their last output, i.e., for all input sequences $\trace\in (2^{I})^*$, $\tree(\trace) = \stratoutput(\moore,\trace)$.

\medskip
\noindent\textit{$\omega$-Automata.}
A (deterministic) $\omega$-\emph{automaton} $\aut$ (over alphabet $\Sigma$) is a tuple $(Q,q_0,\delta,\Omega)$ consisting of a finite set of states $Q$, an initial state $q_0\in Q$,
a transition function $\delta \colon Q\times \Sigma \rightarrow Q$, and an acceptance condition $\acc\subseteq Q^\omega$.
The unique \emph{run} of $\aut$ from state $q$ on some trace $\trace\in \Sigma^\infty$, denoted by $\runt(\aut,q,\trace)$, is a sequence of states $\run\in Q^\infty$ with $\abs{\run} =\abs{\trace}+1$, $\run[0] = q$, and $\delta(\run[i],\trace[i]) = \run[i+1]$ for all $i\in\dom(\trace)$. 
A run $\run$ is \emph{accepting} if $\run\in \acc$.
The language $\mathcal{L}(\aut)$ is the set of all traces $\trace$ for which the unique run $\runt(\aut,q_0,\trace)$ is accepting.

A (deterministic) safety automaton (DSA) is an $\omega$-automaton with a safety acceptance condition $\acc = \safety(F)$, where $F\subseteq Q$ is a set of safe states, containing the set of runs that only visits states in $F$.
It is known that an LTL specification $\varphi$ can be translated into an equivalent $\omega$-automaton $\aut$, i.e., with $\lang(\varphi) = \lang(\aut)$, with a double exponential blow-up~\cite{baier2008principles}.
A safety LTL specification $\varphi$ is an LTL specification for which the corresponding automaton $\aut$ is a safety automaton.
We write $\ltlToAut$ to refer to the procedure that implements this translation from a safety LTL specification to an equivalent safety automaton~\cite{safeLtl,DBLP:conf/cav/KupfermanV99}.

Given an automaton $\aut = (Q,q_0,\delta,\acc)$ and a Moore machine $\moore$ over inputs $I$ and outputs $O$ with $2^{I\cup O}\subseteq \Sigma$, 
we define the composition $\aut \times \moore = (Q\times S, (q_0,s_0), \delta')$ as a product automaton with partial transition function $\delta': (Q\times S)\times 2^I \rightarrow (Q\times S)$ defined as $\delta'((q,s),\sigma) = (\delta(q,\sigma\cup\labelts(s)),\tau(s,\sigma))$.
We write $\runs(\aut,q,\moore)$ to denote the set of runs $\run$ of $\aut$ starting from $q$ for which there exists a corresponding run $\run'$ of $\aut\times\moore$ with $\run'[i] = (\run[i],s_i)$ for each $i\geq 0$.
We write $\reach(\aut\times\moore)$ to denote the set of states $(q,s)$ that are reachable from the initial state $(q_0,s_0)$ in $\aut\times\moore$, i.e., there exists a run $\run$ of $\aut\times\moore$ with $\run[0] = (q_0,s_0)$ and $\run[k] = (q,s)$ for some $k\geq 0$.

\medskip
\noindent\textit{Tree Automata.}
A tree automaton over $2^{O}$-labeled $2^{I}$-trees is a tuple $\aut = (Q,q_0,T,\acc)$ consisting of a finite set of states $Q$, an initial state $q_0\in Q$, a transition function $T\subseteq Q\times 2^{O} \times (2^{I} \rightarrow Q)$, and an acceptance condition $\acc\subseteq Q^\omega$.

A run of a tree automaton $\aut$ on a tree $\tree: (2^I)^*\rightarrow 2^O$ is a $Q$-labeled $2^{I}$-tree $\runtree : (2^{I})^* \rightarrow Q$ with $\runtree(\epsilon) = q_0$ and for every input sequence $\trace\alpha \in (2^{I})^*$, there exists $(\runtree(\trace), \tree(\trace), f) \in T$ such that $f(\alpha) = \runtree(\trace\alpha)$. 
Given an input sequence $\alpha_0\alpha_1\cdots\in (2^{I})^*$, a run (tree) $\runtree : (2^{I})^* \rightarrow Q$ produces a path, i.e., a sequence of states, $\runtree(\epsilon)\runtree(\alpha_0)\runtree(\alpha_0\alpha_1)\cdots$.
The run (tree) is accepting if every path produced by the run is accepting by the acceptance condition $\acc$.
We define the language of a tree automaton $\lang(\aut)$ as the set of all accepting trees.

\section{Universal Controllers}\label{sec:general-controller}
Logical controller synthesis is the problem of, given a plant and an environment, finding a controller that satisfies, together with the plant, some specification.
More specifically, for a given architecture, an LTL specification $\varphi$ over $\ap$, and a plant strategy $\strat_\plant$, the problem asks to find a controller strategy $\strat_\contr$ s.t.\ $\strat_\contr \parallel \strat_\plant \vDash \varphi$. 
We formalize the example from \Cref{example:running} as follows:

\begin{figure}[t]
    \begin{subfigure}[t]{.5\textwidth}
        \centering
        \tikzstyle{state}=[draw, circle, fill=none, minimum width=0.7cm, 
minimum height = 0.7cm,
align=center, thick]

\begin{tikzpicture}[->,>=stealth',shorten >= 1pt,auto]

\node[state] (p0) at (0,0){%
    $~q_0$
};
\node[state] (p1) at (2,1.5){%
    $~q_1$
};
\node[state] (p2) at (3,0){%
    $~q_2$
};

\node[state, draw=none] (init)[left = 0.7 of p0]{%
};

\path (init) edge[thick] (p0)
(p0) edge[thick,bend left=20] node[above left=-0.1] {%
  $\neg o_\env$
  } (p1)
(p0) edge[loop above,thick] node[left] {%
  $o_\env \!\wedge\! (o_\plant\! \leftrightarrow\! o_\contr)$
  } ()
(p1) edge[loop right,thick] node[right] {%
  $\true$
  } ()
(p0) edge[thick] node[above] {%
  $o_\env \!\wedge \!(\neg o_\plant \!\leftrightarrow\! o_\contr)$
  } (p2)
(p2) edge[loop right,thick] node[above=0.5em] {%
  $\true$
  } ()
;

\end{tikzpicture}
    \end{subfigure}
    \hspace*{0.5cm}
    \begin{subfigure}[t]{.4\textwidth}
        \centering
        \tikzstyle{state}=[draw, circle, fill=none, minimum width=0.7cm, 
minimum height = 0.7cm,
align=center, thick]

\begin{tikzpicture}[->,>=stealth',shorten >= 1pt,auto]

\node[state] (p0){%
    $~s_0$
};
\node[state] (p1) [right = 1 of p0]{%
    $~s_1$
};

\node [below=0.1em of p0] (l0) {$\textcolor{labels}{o_\plant}$};
\node [below=0.1em of p1] (l1) {$\textcolor{labels}{\neg o_\plant}$};

\node[state, draw=none] (init)[left = 0.7 of p0]{%
};

\path (init) edge[thick] (p0)
(p0) edge[bend left=20,thick] node[above] {%
  $\neg o_\env$
  } (p1)
(p1) edge[bend left=20,thick] node[below] {%
  $o_\env$
  } (p0)
(p0) edge[loop above,thick] node[above] {%
  $o_\env$
  } ()
(p1) edge[loop above,thick] node[above] {%
  $\neg o_\env$
  } ()

      ;
\end{tikzpicture}
    \end{subfigure}   
    \vspace{-0.8cm}
    \caption{Left: the (safety) automaton~$\aut_\varphi=\ltlToAut(\varphi)$ with safe states $F = \{q_0,q_1\}$ representing the specification $\varphi$ in~\eqref{equ:spec:exp}. Right: a plant strategy $\strat_\plant$.}
    \label{fig:running:example}
    \label{fig:running:automaton}
    \label{fig:running:plant}
\end{figure}

\begin{example}\label{example:automaton}
Consider the architecture over $\ap = \{o_\contr,o_\plant,o_\env\}$ with $\outputs{\env} = \{o_\env\}$, $\outputs{\contr} = \{o_\contr\}$, $\outputs{\plant} = \{o_\plant\}$, and  LTL specification $\varphi$ in~\Cref{equ:spec:exp}.
As discussed in \Cref{example:running}, the specification requires that the controller needs to match the output of the plant as long as the environment outputs $o_\env$.
This is also illustrated by the automaton $\aut_\varphi$ depicted in \cref{fig:running:automaton} (left) which accepts $\mathcal{L}(\varphi)$ over the alphabet $\Sigma = 2^{\ap}$.
We want to synthesize a controller strategy $\strat_\contr$ that is correct w.r.t.\ a given plant strategy $\strat_\plant$, i.e., the specification is realizable by the composition of the plant and the controller.
\end{example}

In general, infinitely many plants exist for which one can find a controller such that the specification over the overall architecture becomes realizable. 
We call such plants \emph{admissible}. 

\begin{definition}[Admissibility]\label{def:admissibility}
    Given a specification $\varphi$, a plant $\strat_\plant$ is said to be \emph{admissible} for $\varphi$ if there exists a controller $\strat_\contr$ s.t. $\strat_\contr \parallel \strat_\plant \vDash \varphi$. 
\end{definition}

Given a plant strategy $\strat_\plant$ and a specification $\varphi$, the \emph{controller synthesis problem} asks to construct a controller strategy $\strat_\contr$ s.t. $\strat_\contr \parallel \strat_\plant \vDash \varphi$, or determine that no such controller exists, i.e., to conclude that the synthesis problem is \emph{unrealizable}. The latter is true iff the plant $\strat_\plant$ is not admissible.
Given a plant strategy $\strat_\plant$, typical controller synthesis techniques explore the entire state space of $\strat_\plant$ to either construct a controller or conclude unrealizability. In practice, such plant representations can be very large, as they might represent abstractions of actual implementations or physical system dynamics with a large number of states and propositions. This quickly leads to a large and infeasible runtime of synthesis algorithms.

The major observation of this work is that this exploration can be circumvented by solving the controller synthesis problem \enquote{abstractly} without considering a particular plant strategy. We achieve this by generating future assumptions about the plant. We call such future assumptions \emph{prophecies}. 

\begin{definition}[Prophecies]\label{def:prophecy}
    A prophecy $\prophecy\subseteq\plantall$ is a set of plants. 
\end{definition}

Intuitively, we use prophecies to condition the current controller output on the future behavior of the plant strategy. 
To illustrate this, consider again the controller synthesis problem from \Cref{example:running}. 
As future behavior of the environment is unknown,
for any admissible plant, the controller needs to output $o_\contr$ in the first time-step if the plant also outputs $o_\plant$ in the first time-step, else it needs to output $\neg o_\contr$.
However, the controller cannot observe the plant output before it outputs $o_\contr$ or $\neg o_\contr$. Hence, the controller needs to use a prophecy about the plant to decide on its output.
Note that the plant strategy can be very complex, involving a lot of other behavior and propositions. By synthesizing a controller which conditions its output on a \emph{prophecy}, we extract the information about the future of the plant relevant for controller synthesis for a given specification and architecture, independently of the concrete plant realization. When given a concrete plant strategy, we then verify all prophecies to pick the correct controller output in every time step. We call such a prophecy-annotated controller a \emph{universal controller}, as formalized next.

\begin{definition}[Universal controller]
A \emph{universal controller} $\univcontr$ over alphabet $\Sigma$ and prophecy set $\prophecyall$ is a tuple $(S, s_0,\tau,\kappa)$ consisting of a finite set of states $S$, an initial state $s_0 \in S$, a transition function $\tau: S \times \Sigma \rightarrow S$, and a prophecy annotation $\kappa: S \times 2^{\outputs{\contr}} \rightarrow \prophecyall$. 
\end{definition}

Note that the transition function of a universal controller is over the whole alphabet~$\Sigma$, instead of just inputs of the controller.
Furthermore, each state is annotated with an additional \emph{condition} by a prophecy for each controller output.
Once we obtain a plant, we can verify these prophecies to pick the correct controller output for each state and set the labelling function of the controller accordingly.
Once the labelling function is set, this also refines the transition function into a function over the inputs only.

\begin{figure}[t]
    \newcommand{\arrow}{\twoheadleftarrow}
    \centering
    \begin{subfigure}[c]{.59\textwidth}
        \tikzstyle{state}=[draw, circle, fill=none, minimum width=0.7cm, 
minimum height = 0.7cm,
align=center, thick]

\begin{tikzpicture}[->,>=stealth',shorten >= 1pt,auto]

\node[state] (p0) at (0,0){%
    $~q_0$
};
\node[state] (p1) at (2,1.5){%
    $~q_1$
};
\node[state] (p2) at (3,0){%
    $~q_2$
};

\node[draw=none] (init)[left = 0.7 of p0]{%
};

\path (init) edge[thick] (p0)
(p0) edge[thick,bend left=20] node[above left=-0.1] {%
  $\neg o_\env$
  } (p1)
(p0) edge[loop above,thick] node[left] {%
  $o_\env \!\wedge\! (o_\plant\! \leftrightarrow\! o_\contr)$
  } ()
(p1) edge[loop right,thick] node[right] {%
  $\true$
  } ()
(p0) edge[thick] node[above] {%
  $o_\env \!\wedge \!(\neg o_\plant \!\leftrightarrow\! o_\contr)$
  } (p2)
(p2) edge[loop right,thick] node[above=0.5em] {%
  $\true$
  } ()
;

\node [below=0.2em of p0] (l0) {$\textcolor{labels}{~~o_\contr}\arrow \bracket{o_\plant}$};
\node [above=0.2em of p1] (l1) {$\textcolor{labels}{~~~~~o_\contr}\arrow\bracket{\true}$};
\node [below=0.2em of p2] (l2) {$\textcolor{labels}{~~o_\contr}\arrow\bracket{\false}$};
\node [below=1.2em of p0] (l0) {~$\textcolor{labels}{\neg o_\contr}\arrow \bracket{\neg o_\plant}$};
\node [above=1.2em of p1] (l1) {$\textcolor{labels}{~~~\neg o_\contr}\arrow\bracket{\true}$};
\node [below=1.2em of p2] (l2) {$\textcolor{labels}{\neg o_\contr}\arrow\bracket{\false}$};

\end{tikzpicture}
    \end{subfigure}
    \begin{subfigure}[c]{.39\textwidth}
        \tikzstyle{state}=[draw, circle, fill=none, minimum width=0.7cm, 
minimum height = 0.7cm,
align=center, thick]

\begin{tikzpicture}[->,>=stealth',shorten >= 1pt,auto]

\node[state] (p0){%
    $~q_0s_0$
};
\node[state] (p1) [right = 1 of p0]{%
    $~q_1s_1$
};

\node[draw=none] (init)[left = 0.7 of p0]{%
};

\node [below=0.1em of p0] (l0) {$\textcolor{labels}{o_\contr}$};
\node [below=0.1em of p1] (l1) {$\textcolor{labels}{\neg o_\contr}$};

\path (init) edge[thick] (p0)
(p0) edge[thick] node[above] {%
  $\neg o_\env$
  } (p1)
(p0) edge[loop above,thick] node[above] {%
  $o_\env$
  } ()
(p1) edge[loop above,thick] node[above] {%
  $\true$
  } ()

      ;
\end{tikzpicture}
    \end{subfigure}
    \vspace{-0.8cm}      
    \caption{Left: a universal controller $\univcontr$ with prophecy annotation $\kappa$ (shown by $\arrow$) such that 
    $\kappa(q_1,*) = \bracket{\true}$ (i.e., the set of all plants),
    $\kappa(q_2,*) = \bracket{\false}$ (i.e., the empty set), 
    $\kappa(q_0,o_\contr) = \bracket{o_\plant}$ (i.e., the set of all plants that output $o_\plant$ in the current time-step), and
    $\kappa(q_0,\neg o_\contr) = \bracket{\neg o_\plant}$ (i.e., the set of all plants that output $\neg o_\plant$ in the current time-step).
    Right: a controller strategy $\strat_\contr$ consistent with this universal controller w.r.t.\ the plant shown in \Cref{fig:running:plant} (right).} 
    \label{fig:running:universal}
    \label{fig:running:controller}
\end{figure}

\begin{example}
Consider again the controller synthesis problem from \Cref{example:running} and a universal controller $\univcontr$ as shown in \Cref{fig:running:universal} (left). 
In the first time-step, as future behavior of the environment is unknown, the controller must set $o_\contr$ to true if
the plant also sets $o_\plant$ to true.
Hence, the prophecy for the output $o_\contr$ at the initial state $q_0$ in \Cref{fig:running:universal} (left), i.e., $\kappa(q_0,o_\contr) = \bracket{o_\plant}$, is  the set of all plants that output $o_\plant$ in the current time-step.
Similarly, the prophecy for the output $\neg o_\contr$ at state $q_0$, i.e., $\kappa(q_0,\neg o_\contr) = \bracket{\neg o_\plant}$, is the set of all plants that output $\neg o_\plant$ in the current time-step.
Furthermore, once the environment outputs $\neg o_\env$, we reach the state $q_1$ where the controller can set any output as the specification is satisfied no matter what the controller does. The prophecy for both outputs $o_\contr$ and $\neg o_\contr$ at state $q_1$ is $\bracket{\true}$, i.e., it can set either output for any plant.
Similarly, if the environment outputs $o_\env$, and the output of the plant and the controller does not match, we reach the state $q_2$ where the specification is violated no matter what the controller does. Hence, the prophecy for both output $o_\contr$ and $\neg o_\contr$ at state $q_2$ is $\bracket{\false}$, i.e., there is no output that is correct for any plant.
\end{example}

Such a universal controller is \emph{universal} in the sense that it represents a set of controllers for each plant as formalized next.
\begin{definition}[Consistency]\label{def:consistency}
    Given a universal controller $\univcontr = (S, s_0,\tau,\kappa)$ and a plant $\strat_\plant$, a controller $\strat_\contr = (S^\contr, s_0^\contr,\tau^\contr,\labelts^\contr)$ is said to be \emph{consistent} with $\univcontr$ w.r.t. $\strat_\plant$, denoted by $\strat_\contr\vDash\univcontr\parallel \strat_\plant$, 
    if for all $(s,s^\plant,s^\contr)\in \reach(\univcontr\times(\strat_\plant\parallel\strat_\contr))$, it holds that $\strat_\plant(s^\plant)\in \kappa(s,\labelts^\contr(s^\contr))$.
\end{definition}
Intuitively, at each state of the universal controller, the corresponding plant state satisfies the prophecy for the controller output from that state.
Given a universal controller $\univcontr$ and a plant $\strat_\plant$, one can find a controller $\strat_\contr$ such that $\strat_\contr\vDash\univcontr\parallel \strat_\plant$ (if one exists) by verifying the plant against each prophecy in the current state of the universal controller and
picking a controller output for which the corresponding prophecy is 
satisfied by the plant.

\begin{example}
    Consider the universal controller $\univcontr$ shown in \Cref{fig:running:universal} (left) and the plant strategy $\strat_\plant$ shown in \Cref{fig:running:plant} (right).
    We can construct a controller strategy $\strat_\contr$ that is consistent with $\univcontr$ w.r.t. $\strat_\plant$ by verifying the prophecies as follows.
    From the initial state $q_0$, the plant $\strat_\plant$ sets $o_\plant$ to true and hence, satisfies the prophecy $\kappa(q_0, o_\contr)$. Thus, a consistent controller has to set $o_\contr$ to true in the first time-step.
    Later on from state $q_0$, the plant $\strat_\plant$ continues setting $o_\plant$ to true as long as the environment outputs $o_\env$.
    Hence, a consistent controller also has to set $o_\contr$ to true as long as the environment outputs $o_\env$.
    As soon as the environment outputs $\neg o_\env$, we reach state $q_1$ in the universal controller $\univcontr$ from which the prophecy for $o_\contr$ and $\neg o_\contr$ are $\bracket{\true}$, and hence, the controller can set any output.
    By picking the prophecy $\kappa(q_1,\neg o_\contr)$, the controller $\strat_\contr$ sets $o_\contr$ to false for all future time-steps. 
    With this, we obtain the controller strategy $\strat_\contr$ shown in \Cref{fig:running:controller} (right) which is consistent with $\univcontr$ w.r.t. $\strat_\plant$.
    Note that as we set controller outputs by verifying the prophecies for the plant $\strat_\plant$, we never had to verify prophecies at the state $q_2$ in the universal controller $\univcontr$ where the specification is violated.
\end{example}

Given the above representation of controllers by universal controllers, our main goal is to find a universal controller such that given any admissible plant, the controller strategy consistent with them solves the logical controller synthesis problem. We call such a universal controller \emph{correct} as formalized below.

\begin{definition}[Correctness]\label{def:correctness:of:universal}
    Given a specification~$\varphi$, a universal controller $\hyperstrat$ is \emph{correct} for $\varphi$ if for all admissible plant strategies $\strat_\plant$, it holds that $\strat_\contr\vDash \univcontr\parallel \strat_\plant$ implies 
    $\strat_\plant \parallel \strat_\contr \vDash \varphi$.
\end{definition}

Conversely, if all correct controllers for an admissible plant are consistent with the universal controller, we call such a universal controller \emph{most permissive}.

\begin{definition}[Permissiveness]\label{def:permissive:of:universal}
    Given a specification~$\varphi$, a universal controller $\hyperstrat$ is \emph{most permissive} for $\varphi$ if for all admissible plant strategies $\strat_\plant$, it holds that $\strat_\plant \parallel \strat_\contr \vDash \varphi$ implies $\strat_\contr\vDash \univcontr\parallel \strat_\plant$.
\end{definition}

\begin{example}
    Consider the universal controller $\univcontr$ shown in \Cref{fig:running:universal} (left), the plant strategy $\strat_\plant$ shown in \Cref{fig:running:plant} (right), and the controller strategy $\strat_\contr$ shown in \Cref{fig:running:controller} (right) that is consistent with $\univcontr$ w.r.t. $\strat_\plant$.
    The universal controller $\univcontr$ is correct and most permissive for the specification $\varphi$ in~\eqref{equ:spec:exp} and hence, the controller strategy $\strat_\contr$ solves the logical controller synthesis problem for $\varphi$ with given plant strategy~$\strat_\plant$.
\end{example}

We define the formal problem statement as follows:

\begin{definition}[Universal controller synthesis]\label{prob:general:controller}
    Given a specification $\varphi$, the problem of \emph{universal controller synthesis} is to find a universal controller that is correct and most permissive for $\varphi$.
\end{definition}

Computing a universal controller reduces to the problem of computing suitable prophecies. 
The following section presents an automata-based algorithm to compute prophecies for safety specifications. These prophecies build the foundation of our algorithm for synthesizing universal controllers.

\section{Prophecies for Controller Synthesis}\label{sec:propehcies}
In this section, we introduce a class of prophecies, namely \emph{safe prophecies}, which are sufficient to construct a universal controller for safety specifications that is both correct and most permissive.

\subsection{Prophecies}
The class of safe prophecies only considers the current and next state of execution to capture the safety part of the given specification.
Intuitively, for the current choice of the controller output, a safe prophecy collects a set of plants for which this choice is \emph{safe}, i.e., it can lead to an accepting run in the automaton. 
We define these prophecies for a given automaton:

\begin{definition}\label{def:safe:prophecies}
    Let $\aut = (Q,q_0,\delta,\acc)$ be an automaton, $q\in Q$ be a state in $\aut$,  and $\alpha\in 2^{\outputs{\contr}}$ be an output of the controller. 
    The \emph{safe prophecy} $\safeProphecy_\aut^{q,\alpha}$ is the set
    \[ \safeProphecyF{\aut}{q,\alpha} = \{\strat_\plant \in \plantall \mid  \exists \strat_\contr\in\contrall .\ \stratoutput(\strat_\contr,\epsilon) = \alpha \wedge \runs(\aut,q,\strat_\plant\parallel\strat_\contr)\subseteq\acc\}.\]
\end{definition}

Given a state $q\in Q$ and a controller output $\alpha\in 2^{\outputs{\contr}}$, the set $\safeProphecy_\aut^{q,\alpha}$ collects all plant strategies for which there exists a controller strategy $\strat_\contr$ from $q$ that is consistent with the current output $\alpha$ and produces only accepting runs with $\strat_\plant$ in the automaton, i.e., $\runs(\aut,q,\strat_\plant\parallel\strat_\contr)\subseteq\acc$.
Intuitively, the safe prophecy of a controller's choice collects all the plant strategies for which this choice is safe.
Hence, for any safety LTL formula, safe prophecies provide necessary restrictions on the controller's choices.

\begin{restatable}{theorem}{restateNecessary}\label{thm:safeprophecy-necessary}
    Let $\varphi$ be a safety LTL specification and $\aut = (Q,q_0,\delta,\acc = \safety(F))$ its equivalent safety automaton.  
    The universal controller $\univcontr = (Q,q_0,\delta,\kappa)$ with $\kappa(q,\alpha) = \safeProphecyF{\aut}{q,\alpha}$ is most permissive for $\varphi$.
\end{restatable}
\begin{proof}
    Let $\strat_\plant = (S^\plant, s_0^\plant,\tau^\plant,\labelts^\plant)$ be an admissible plant and let the controller strategy $\strat_\contr = (S^\contr, s_0^\contr,\tau^\contr,\labelts^\contr)$ be such that $\strat_\plant\parallel\strat_\contr\vDash\varphi$. We need to show that $\strat_\contr\vDash\univcontr\parallel\strat_\plant$.
    Consider the composition $\univcontr' = \univcontr\times (\strat_\plant\parallel\strat_\contr) = (S', s_0', \tau')$.
    Let $s' = (q,s^\plant,s^\contr)\in\reach(\univcontr')$.
    Then, it is enough to show that $\strat_\plant(s^\plant)\in\kappa(q,\labelts^\contr(s^\contr)) = \safeProphecyF{\aut}{q,\labelts^\contr(s^\contr)}$.
    This means, it is enough to show that for strategy $\strat_\contr' = \strat_\contr(s^\contr)$ and $\strat_\plant' = \strat_\plant(s^\plant)$, it holds $\runs(\aut,q,\strat_\plant' \parallel \strat_\contr')\subseteq\acc$.
    
    Let $\run$ be a run in $\runs(\aut,q,\strat_\plant' \parallel \strat_\contr')$, then by definition, there exists a run $\run'$ in $\univcontr'$ with $\run'[i] = (\run[i],*,*)$ for all $i\geq 0$.
    Let $\run^f = (q_0,s_0^\plant,s_0^\contr)\dotsc(q_k,s_k^\plant,s_k^\contr)$ be a path in $\univcontr'$ from $s_0'$ to $s'$.
    Then, by construction, $(\run^f[0,k-1])\run'$ is a run from $s_0'$ in $\univcontr'$.
    As $\strat_\contr\parallel\strat_\plant\vDash\varphi$ and $\lang(\aut)=\lang(\varphi)$, the corresponding run of $(\run^f[0,k-1])\run'$ in $\aut$ is accepting, i.e., $q_0q_1\dotsc q_{k-1}\run\in\acc$. By the definition of safety acceptance, we have $\run\in\acc$.
    As $\run$ is arbitrary, we have $\runs(\aut,q,\strat_\plant' \parallel \strat_\contr') \subseteq \acc$, and hence, $\strat_\contr\vDash\univcontr\parallel\strat_\plant$.
\end{proof}

The theorem shows that every controller strategy that is correct for a plant is preserved in the prophecy set.
Furthermore, as ensuring that the controller does not make any unsafe choices is sufficient to realize a safety specification, we next show that the safe prophecies are sufficient to synthesize a correct universal controller realizing any safety specification once a plant strategy is given.

\begin{restatable}{theorem}{restateSufficient}\label{thm:safeprophecy-sufficient}
    Let $\varphi$ be a safety LTL specification and $\aut = (Q,q_0,\delta,\acc = \safety(F))$ its equivalent safety automaton.  
    The universal controller $\univcontr = (Q,q_0,\delta,\kappa)$ with $\kappa(q,\alpha) = \safeProphecyF{\aut}{q,\alpha}$ is correct for $\varphi$.
\end{restatable}

\begin{proof}
    Let $\strat_\plant = (S^\plant, s_0^\plant,\tau^\plant,\labelts^\plant)$ be an admissible plant and let the controller strategy $\strat_\contr = (S^\contr, s_0^\contr,\tau^\contr,\labelts^\contr)$ be consistent with $\univcontr$ w.r.t. $\strat_\plant$, i.e., $\strat_\contr\vDash \univcontr\parallel\strat_\plant$.
    We need to show that $\strat_\plant\parallel\strat_\contr\vDash\varphi$.
    Let $\trace\in\traces(\strat_\plant\parallel\strat_\contr)$ with $\run = \runt(\aut,q_0,\trace)$ being the corresponding unique run.
    As $\lang(\aut) = \lang(\varphi)$, to prove $\trace\in\lang(\varphi)$, we only need to show that $\run\in\acc$, i.e., $\run[i]\in F$ for all $i\geq 0$.
    
    Let $\strat_\plant = (S^\plant, s_0^\plant,\tau^\plant,\labelts^\plant)$ and $\strat_\contr = (S^\contr, s_0^\contr,\tau^\contr,\labelts^\contr)$.
    Then, as $\strat_\contr\vDash \univcontr\parallel\strat_\plant$, by definition, for every $(q,s^\plant,s^\contr)\in \reach(\univcontr\times(\strat_\plant\parallel\strat_\contr))$, it holds that $\strat_\plant(s^\plant)\in\kappa(q,\labelts^\contr(s^\contr)) = \safeProphecyF{\aut}{q,\labelts^\contr(s^\contr)}$.
    
    As $\trace\in\traces(\strat_\plant\parallel\strat_\contr)$, by construction, there exists a run $\run'$ from the initial state in $\univcontr\times(\strat_\plant\parallel\strat_\contr)$ corresponding to $\trace$ such that $\run'[i] = (\run[i],s_i^\plant,s_i^\contr)$.
    Hence, for all $i\geq 0$, we have
    $\run'[i]\in\reach(\univcontr\times(\strat_\plant\parallel\strat_\contr))$, and hence, $\strat_\plant(s_i^\plant)\in\safeProphecyF{\aut}{\run[i],\labelts^\contr(s_i^\contr)}$.
    Then, by definition of safe prophecies, for each $i\geq 0$, there exists $\strat_\contr^i\in\contrall$ with $\stratoutput(\strat_\contr^i,\epsilon) = \labelts^\contr(s_i^\contr)$ and $\runs(\aut,\run[i],\strat_\contr^i\parallel\strat_\plant(s_i^\plant))\subseteq\acc$.
    This means, for each $i\geq 0$, there exists runs from $\run[i]$ in $\aut$ that are accepting, and hence, by definition of safety acceptance $\acc$, we have $\run[i]\in F$.
\end{proof}

The universal controller for our running example shown in \Cref{fig:running:universal} (left) is annotated with safe prophecies and hence, is correct and most permissive for the safety specification $\varphi$ in~\eqref{equ:spec:exp}.

\subsection{Automata for Prophecies}\label{sec:automata:for:prophecies}
Prophecies, as defined in \Cref{def:prophecy} and \Cref{def:safe:prophecies}, define sets of plants that are correct for an explicit controller choice.
A suitable automaton model for prophecies are tree automata whose languages are sets of trees (in contrast to sets of traces for word automata).
Every tree in the language of such an automaton corresponds to a strategy of a plant, which is also an infinite tree.
We show that we can construct a tree automaton that correctly accepts the safe prophecy for a given controller choice.
The basic idea of the construction is similar to the automaton to two-player game conversion for LTL synthesis~\cite{DBLP:series/natosec/Finkbeiner16}.
A plant satisfies the prophecy if it has a correct reaction to any choice of the controller inputs.
Additionally, we restrict the output of the controller in the first step of the automaton, i.e., the output we compute the prophecy for. 
This yields the following result.

\begin{restatable}{theorem}{restateTree}\label{thm:safe:prophecy:tree}
    Let $\aut = (Q, q_0, \delta, \Omega)$ be a safety automaton for the safety LTL formula $\varphi$ and $q_s \in Q, \alpha_s \in 2^{O_\contr}$ be an output of the controller.
    There exists a linear tree automaton $\aut'$ s.t. $\mathcal{L}(\aut') = \safeProphecy_{\aut}^{q_s,\alpha_s}$. 
\end{restatable}
\begin{proof}
    We construct a tree automaton $\aut'$ over $(2^{O_\plant})$-labelled $2^{I_\plant}$ trees that accepts all plants for which a controller strategy exists with the current choice $\alpha_s$ that ensures the specification. 
    The idea is to simulate the plant and controller behavior on the word automaton $\aut$, where
    the automaton is implicitly labelled with both plant and controller outputs, and only branches over environment outputs.
    To make the tree automaton labelled with plant outputs, 
    we remove each controller output $\beta_\contr$ from the labels and make the corresponding node branch over all possible plant inputs, i.e., outputs of both controller and environment, 
    by taking all branches with inconsistent controller outputs (i.e., $\neq \beta_\contr$) to an accepting sink state.
    In addition, we simulate $\aut$ from $q_s$ and restrict the initial choice of the controller to $\alpha_s$ by introducing a new initial state $q_0'$.
    Formally, the automaton $\aut' =(Q', q_0', T', \Omega')$ is defined as follows:
    \begin{itemize}
        \item $Q' = Q \uplus \{q_0',q_{sink}\}$,
        \item $T'\subseteq Q\times 2^{O_\plant} \times (2^{I_\plant} \rightarrow Q)$ is the minimal set such that
        \begin{itemize}
            \item for all $\beta_\plant \in 2^{O_\plant}$, there is a transition $(q_0', \beta_\plant, f)$ such that for all $\beta \in 2^{I_\plant}$,
            $f(\beta) = \begin{cases}
                \delta(q_s, \beta_\plant\cup\beta) & \text{if } \beta\cap\outputs{\contr} = \alpha_s,\\
                q_{sink} & \text{otherwise},
            \end{cases}$
            \item for all $q \in Q$, $\beta_\plant \in 2^{O_\plant}$, and $\beta_\contr \in 2^{O_\contr}$, there is a transition $(q, \beta_\plant,f)$ such that for all $\beta \in 2^{I_\plant}$,
            $f(\beta) = \begin{cases}
                \delta(q, \beta_\plant\cup\beta) & \text{if } \beta\cap\outputs{\contr} = \beta_\contr,\\
                q_{sink} & \text{otherwise},
            \end{cases}$
            \item for all $\beta_\plant \in 2^{O_\plant}$, there is a transition $(q_{sink}, \beta_\plant, f)$ such that $f(\beta) = q_{sink}$ for all $\beta \in 2^{I_\plant}$.
        \end{itemize}
        \item $\Omega' = \safety(F')$ with $F' = F \cup \{q_0',q_{sink}\}$.
    \end{itemize}
    Compared to $\aut$, the automaton $\aut'$ has two new states $q_0'$ and $q_{sink}$.
    For all transitions from $q_0'$, we fix the controller output to be $\alpha_s$ and quantify over all possible plant outputs. Afterwards, we branch over all possible plant inputs but only simulate the automaton $\aut$ when the controller output is $\alpha_s$, otherwise, we go to the sink state $q_{sink}$.
    Similarly, in the general case, we quantify over all possible controller outputs $\beta_\contr$ and plant outputs $\beta_\plant$, but only use the plant output to explicitly label the tree.
    Then, we branch over all possible plant inputs but only simulate the automaton $\aut$ when the controller output is $\beta_\contr$, otherwise, we go to the sink state $q_{sink}$.
    The runs of this automaton are all trees that branch over environment and controller outputs and are labeled with plant outputs that satisfy the prophecy $\safeProphecyF{\aut}{q_s,\alpha_s}$.
    Since the transition function follows a similar construction for reactive synthesis (without a plant) in~\cite{DBLP:series/natosec/Finkbeiner16}, the correctness of the construction follows immediately.
    As $\abs{Q'} = \abs{Q}+2$, the size of $\aut'$ is linear in the size of $\aut$.
\end{proof}

  \begin{wrapfigure}[14]{r}{0.6\textwidth}
    \vspace*{-0.8cm}
    \tikzstyle{state}=[draw, circle, fill=none, minimum width=0.5cm, 
minimum height = 0.5cm,
align=center, thick]
\tikzstyle{rec}=[draw, rectangle, fill=none, text width=0.25cm, 
align=center, thick]

\begin{tikzpicture}[->,>=stealth',shorten >= 1pt,auto]
  \hspace{-0.6cm}

\node[state] (q00) at (0,0){%
    $~q_0'$
};
\node[state] (q0) at (-0.3,1.5){%
    $~q_0$
};
\node[rec] (f1) at (1.5,1){%
    $f_1$
};
\node[rec] (f2) at (1.5,-0.5){%
    $f_2$
};
\node[state] (q1) at (3.5,1){%
    $~q_1$
};
\node[state] (q2) at (3.5,-0.5){%
    $~q_2$
};
\node[rec] (f3) at (6,0.8){%
    $f_3$
};
\node[rec] (f4) at (6,2){%
    $f_4$
};
\node[rec] (f5) at (6,0.2){%
    $f_5$
};
\node[rec] (f6) at (6,-1){%
    $f_6$
};

\node[state, draw=none] (init)[left = 0.4 of q00]{%
};

\path (init) edge[thick] (q00)
(q00) edge[thick,bend left=10] node[below] {%
  $o_\plant$
  } (f1)
(q00) edge[thick] node[below] {%
  $\neg o_\plant$
  } (f2)
(f1) edge[dashed] node[above] {%
  $~o_\contr\!\wedge\! o_\env$
  } (q0)
(f1) edge[dashed] node[above] {%
  $~o_\contr\!\wedge\! \neg o_\env$
  } (q1)
(f2) edge[dashed, bend left=-10] node[below] {%
  $~o_\contr\!\wedge\! o_\env$
  } (q2)
(f2) edge[dashed,bend left=20] node[below,xshift=0.4cm] {%
  $~o_\contr\!\wedge\!\neg o_\env$
  } (q1)
(q1) edge[thick,bend right=10] node[below,yshift=0.1cm] {%
  $\true$
  } (f3)
(q1) edge[thick,bend left=10] node[above] {%
  $\true$
  } (f4)
(f3) edge[dashed,bend right=05] node[above, xshift=0.2cm,yshift=-0.1cm] {%
  $o_\contr$
  } (q1)
(f4) edge[dashed,bend left=05] node[right,xshift=0.1cm,yshift=-0.05cm] {%
  $\neg o_\contr$
  } (q1)
(q2) edge[thick,bend left=10] node[above] {%
  $\true$
  } (f5)
(q2) edge[thick,bend right=10] node[below] {%
  $\true$
  } (f6)
(f5) edge[dashed, bend left=05] node[right,xshift=0.1cm,yshift=-0.1cm] {%
  $o_\contr$
  } (q2)
(f6) edge[dashed, bend right=05] node[above,yshift=-0.1cm,xshift=0.2cm] {%
  $\neg o_\contr$
  } (q2)
;

\end{tikzpicture}
  \vspace*{-0.9cm}
    \caption{A part of the tree automaton (without transitions of $q_0$ and $q_{sink}$) with safe states $F' = \{q_0',q_0,q_1\}$ for prophecy $\safeProphecyF{\aut_\varphi}{q_0,o_\contr} = \bracket{o_\plant}$ of \Cref{fig:running:universal}, where each function $f_i$ shown as dashed edges.
    }
    \label{fig:running:tree-aut}
  \end{wrapfigure}
  
Note that it is not possible to represent safe prophecies with a word automaton as we are looking for an automaton that accepts all plant strategies of any size, therefore infinite trees.
One key insight is that we do not only have to verify the choices of the controller against the behavior of the plant, but, there must also exist a strategy of the controller such that all future behavior of the plant satisfies the specification.
Both of these criteria are captured in the constructed tree automaton.

\begin{example}
    Consider the universal controller in \Cref{fig:running:universal} (left) for the running example.
    A tree automaton $\aut'$ for prophecy $\kappa(q_0,o_\contr) = \safeProphecyF{\aut_\varphi}{q_0,o_\contr}$ is shown in \Cref{fig:running:tree-aut}.
    As we are considering a prophecy from the initial state $q_0$, we have a new initial state $q_0'$ that is a copy of $q_0$.
    From $q_0'$, we fix the controller output to be $o_\contr$ for simulation, and branch over all possible plant outputs. 
    Given plant output $o_\plant$ from $q_0'$, we simulate the automaton $\aut_\varphi$ by function $f_1$ when the controller output is $o_\contr$ (else we go to the sink state $q_{sink}$).
    Hence, it branches to state $q_0$ and $q_1$ depending on the environment outputs.
    Similarly, given plant output $\neg o_\plant$ from $q_0'$, we simulate the automaton $\aut_\varphi$ by function $f_2$ when the controller output is $o_\contr$.
    From every other state, we branch over all possible plant and controller outputs. 
    For example, from state $q_1$, we fix the controller output to be $o_\contr$ (resp. $\neg o_\contr$) and simulate the automaton $\aut_\varphi$ by function $f_3$ (resp. $f_4$) with consistent controller output.
    As every transition in $\aut_\varphi$ from $q_1$ leads to $q_1$, the function $f_3$ maps every input with consistent controller output to $q_1$ (and inconsistent controller output to $q_{sink}$).
\end{example}

Once we have the tree automaton for the prophecy, we can use it to construct an explicit controller consistent with a universal controller by verifying the plant strategies against the tree automaton.
For verifying a plant against such a tree automaton, we check if the infinite tree produced by the plant (strategy) is an element of the language of the tree automaton.
We refer to the procedure that constructs the automaton for the prophecy as \lstinline[style=default,language=custom-lang]|prophecy$(q, \alpha, \aut)$|.

\section{Universal Controller Synthesis}\label{sec:synthesis}
In this section, we present an automata-theoretic approach to synthesizing correct universal controllers from safety LTL formulas.
We also present an algorithm that, given a plant and a universal controller, constructs a controller that is consistent with the plant and satisfies the specification.

\subsection{Universal Controllers}
\begin{wrapfigure}[8]{r}{0.57\textwidth} 
    \vspace{-2.2cm}
    \begin{minipage}{0.57\textwidth}
    \begin{algorithm}[H]
            \caption{Universal Controller Synthesis}
            \label{alg:general:controller}
        \begin{mycode}
let universalController($\varphi$) := 
@@ let $\aut_\varphi = (Q,q_0,\delta,\acc) \gets \ltlToAut(\varphi)$   
@@ let $\kappa \gets []$ 
@@ for $q \in Q$ and $\alpha \in 2^{\outputs{\contr}}$
@@@@ let $\kappa(q,\alpha) \gets$ prophecy($q,\alpha, \aut_\varphi$)
@@ return $\univcontr = (Q,q_0,\delta,\kappa)$\end{mycode}%
    \end{algorithm}
\end{minipage}
\end{wrapfigure}
\Cref{alg:general:controller} shows the main procedure for the construction of a universal controller.
The first step is to construct an equivalent safety automaton $\aut_\varphi$ for the specification $\varphi$.
The state space of the automaton yields the initial state space for the universal controller.
The algorithm continues with iterating over every pair of state and controller output, and builds the respective prophecy with the construction introduced in~\Cref{sec:automata:for:prophecies}.
Once the map $\kappa$ contains all pairs of states and outputs, we build the universal controller $\univcontr$.
The universal controller can be interpreted as follows: For a given environment output and state $q$, some output $\alpha$ is taken for which the plant behaves as specified in the tree automaton $\kappa(q, \alpha)$.
Since the alphabet of the automaton $\aut_\varphi$ ranges over all propositions, the transition function $\delta$ is also defined over the plant outputs. 
During composition, we ensure that the prophecy also agrees with the current output of the plant to ensure consistency.

\bigskip
\noindent\textbf{Subautomata Sharing.}
Instead of constructing a new tree automaton for every state-output pair in $\aut_\varphi$, we build a global tree automaton that contains all possible plants for the specification and only set new initial states and transitions.
The global tree automaton construction is similar to \lstinline[style=default,language=custom-lang]|prophecy$(q, \alpha, \aut)$| but $q$ is the initial state $q_0$ of $\aut_\varphi$ and $\alpha$ is not considered in the construction.
One obtains the tree automaton $\acc_\tree = (Q, q_0, T, \Sigma)$ that accepts all plant and controller pairs that are accepting for the specification.
In the loop of~\Cref{alg:general:controller}, the only construction step is changing the initial state in $\aut_\tree$ to a duplicate of $q$ and adding an initial step from $q$ with $\alpha$ to the transition relation $T$.
For the implementation of this algorithm, as presented in \Cref{sec:experiments}, it suffices to store the new initial state and the controller outputs. 
The modified copy of the global tree automaton is only computed during the verification in \Cref{alg:composition}.

\subsection{Composition}
\begin{wrapfigure}[13]{r}{0.7\textwidth} 
    \vspace{-2.4cm}
    \begin{minipage}{0.7\textwidth}
\begin{algorithm}[H]
    \caption{On-the-fly Composition}
    \label{alg:composition}
\begin{mycode}  
let composition($\univcontr = (Q,q_0,\delta,\kappa)$, $\strat_\plant =  (Q',q_0',\delta',\labelts)$) := 
@@ let s $\gets (q_0, q_0')$, states $\gets$ Set(), edges $\gets$ Set()
@@ let queue $\gets$ Queue(s)
@@ while $(q_c, q_p)$ $\gets$ queue.pop() do
@@@@ let edgelist $\gets$ stepComposition($q_c, q_p, \univcontr, \strat_\plant$) 
@@@@ for ($\sigma, (q_c', q_p')$) in edgelist do
@@@@@@@@ $\aut_\tree \gets \kappa(q_c, \sigma)$
@@@@@@@@ $\strat_\plant' \gets$ modifiedPlant($\strat_\plant, q_p, \sigma$)
@@@@@@@@ if isMember($\strat_\plant', \aut_\tree$) then
@@@@@@@@@@ edges.add($(q_c, q_p), \sigma, (q_c', q_p'$)) 
@@@@@@@@@@ queue.push($q_c', q_p'$) if not states.contains($q_c', q_p'$)
@@@@ states.add($q_c', q_p'$) 
@@ return (states, $(q_0, q_0')$, edges)
\end{mycode}
\end{algorithm}
\end{minipage}
\end{wrapfigure}

Once the universal controller is computed, we can compose it with a plant to obtain a plant-specific controller that is consistent with the plant's behavior.
The algorithm operates on an on-the-fly exploration of the product state-space of the universal controller and the plant and is shown in \Cref{alg:composition}.
The new initial state of the controller is a tuple $(q_0, q_0')$, i.e., the initial states of the universal controller and the plant.
Note that we omit double parentheses for function calls over tuple of universal controller and plant states whenever it is clear from the context.
We use a \lstinline[style=default,language=custom-lang]|while| loop over the states to be explored in the composition as a termination condition. 
The body of the loop that ranges from line 5 to line 12 starts by computing a one-step composition from the current state $(q_c, q_p)$ in \lstinline[style=default,language=custom-lang]|stepComposition($q_c, q_p, \gencontr, \strat_\plant$)| which takes the states, the universal controller, and the plant as input.
The result of this function is a list of edges, i.e., pairs over an assignment of the alphabet and future state, in the composition that are possible to reach.
However, since not every computed edge is correct for the given plant, the prophecy of this explicit edge has to be checked first before adding it to the controller's edges.
To verify the plant, we first obtain the tree automaton $\aut_\tree$ representing the pro\-phe\-cy with $\kappa(q_c, \sigma)$ and change the input plant to a modified version $\strat_\plant'$ that starts in $q_p$ with $\sigma$.
The membership check is then performed in line 9 with a call to \lstinline[style=default,language=custom-lang]|isMember($\strat_\plant', \aut_\tree$)| returning true if the tree produced by $\strat_\plant$ is an accepting run of $\aut_\tree$.
In case of the membership check being successful, the edge is a correct edge in the controller and is added to the set of edges and, if not already explored, the target of the edge is added to the queue.
The result of the algorithm is then all explored states and edges with the starting state $(q_0, q_0')$.

\smallskip
\noindent\textbf{Tree Automata Checking.}
The most involved computation is verifying the modified plant against the prophecy tree automaton.
In \Cref{alg:composition}, this check is performed in the function \lstinline[style=default,language=custom-lang]|isMember($\strat_\plant', \aut_\tree$)|.
Following the literature~\cite{DBLP:conf/dagstuhl/Kirsten01}, we encode this function into a 2-player game where the environment player chooses the inputs of the plant and the system player has to provide a correct output for every choice of the environment.
If the system player wins, then the infinite tree produced by the plant's Moore machine is contained in the language of the tree automaton.
The winning condition of the game is \emph{safety}, implied by the acceptance condition of the tree automaton, and the state space is linear in the size of the tree automaton and the plant.

\smallskip
\noindent\textbf{Hashing.}
The number of automaton checks of~\Cref{alg:composition} is $\abs{\Sigma} \times |\gencontr| \times |\strat_\plant|$, i.e., 
the product of the alphabet size, the universal controller size, and the plant size.
We reduce the number of automata checks drastically by hashing intermediate results: For a solved game we obtain a set of \emph{winning} and a set of \emph{losing} states in the game.
These states can be directly mapped back to the composition such that we maintain a function \lstinline[style=default,language=custom-lang]|h($\sigma, (q_c, q_p)$)| that stores for already computed edge-state pairs if the game is winning or losing.
Then, before calling \lstinline[style=default,language=custom-lang]|isMember|, we check if the current edge was hashed before and, if so, return the hashed result.
In our experiments in \Cref{sec:experiments}, we observe that the number of \emph{real} containment checks that are computed without hashing is \emph{one} for almost all benchmarks.
This is the case if every state is reachable for every initial transition from ($q_0, q_0')$ since this is the first state that we compute the safety game for.

\smallskip
\noindent\textbf{Early Termination.}
The local exploration of the state space allows the use of heuristics for termination during both, composition of universal controller with plant and construction of the game.
One such heuristics is defined over the transition function of currently discovered edges: Once the explored set of edges has a transition for every possible output of the environment, then we can stop exploring the current edges and move on to the next state in the queue. 
We do this check before jumping from line 11 to line 6.
Furthermore, we order the edgelist according to already explored states so that we look for such edges before verifying edges to new states.
The semantics of the specification is used to reduce the size of the game to be solved during verification.
Since the acceptance condition is solely based on the specification automaton, we can stop exploring the state space of the plant once we are in a clearly accepting/rejecting loop, e.g., a terminal state.
The early termination heuristics all benefit from one key observation: the specifics of the plant implementation can often be ignored while solving the controller synthesis problem.

Note that for simplicity, we construct the controller as a Mealy machine, i.e., as usual for automata, we use edge labels instead of state labels.
However, this can easily be transformed into Moore machines by ignoring the first environment output and moving the controller output to the source state of every edge.

\subsection{Correctness and Complexity}
\Cref{alg:general:controller} and \Cref{alg:composition} together compute a controller for a given specification and plant with the automata constructions presented in \Cref{sec:automata:for:prophecies}.
The correctness of the algorithm immediately follows from (\cref{thm:safeprophecy-sufficient}).
As the tree automata for the prophecies are linear in the size of the deterministic specification automaton, solving the game in polynomial time results in the same complexity as the standard controller synthesis problem~\cite{DBLP:series/natosec/Finkbeiner16}. Moreover, as the number of prophecies is linear in the size of the specification automaton, the overall time complexity of universal controller synthesis remains doubly exponential in the size of the specification.
In summary, we observe the following result:

\begin{corollary}\label{corollary:correctness-algo}
    Let $\varphi$ be a safety LTL formula.
    \Cref{alg:general:controller} returns a universal controller $\gencontr$ that is correct and most permissive for $\varphi$.
    Furthermore, given an admissible plant $\strat_\plant$, \Cref{alg:composition} returns a controller $\strat_\contr$ such that $\strat_\contr \parallel \strat_\plant \vDash\varphi$.
    The overall time complexity of computing $\strat_\contr$ using \Cref{alg:general:controller} and \Cref{alg:composition} is doubly-exponential in the size of the specification $\varphi$ and linear in the size of the plant~$\strat_\plant$.
\end{corollary}

\begin{figure}[t]
    \centering
    \begin{subfigure}[t]{.48\textwidth}
        \includegraphics[width=0.9\textwidth]{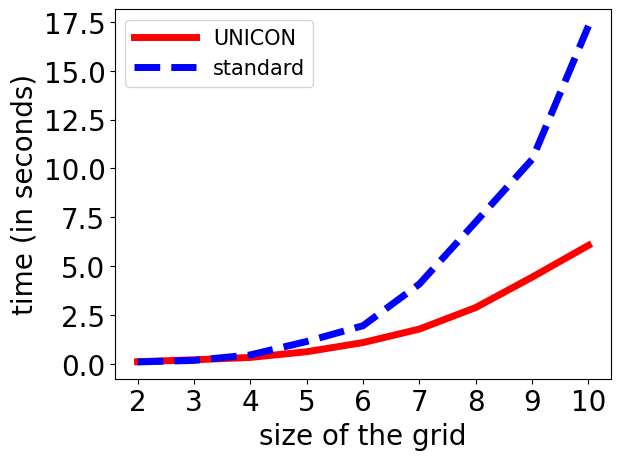}
        \caption{Specification Type 1}
        \label{fig:spec1}
    \end{subfigure}%
    \hfill
    \begin{subfigure}[t]{.48\textwidth}
        \includegraphics[width=0.9\textwidth]{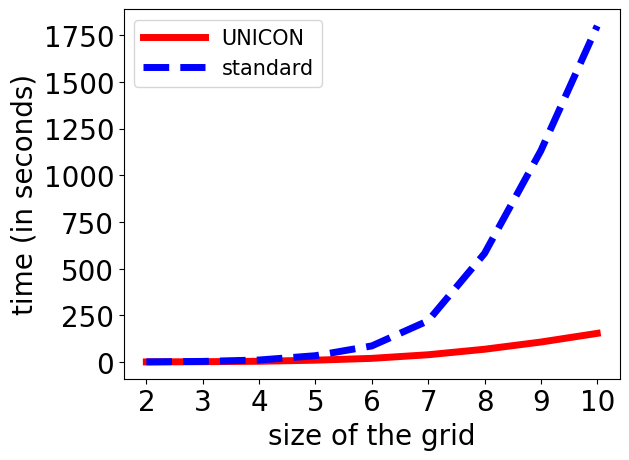}
        \caption{Specification Type 2}
        \label{fig:spec2}
    \end{subfigure}%

    \begin{subfigure}[b]{.48\textwidth}
        \includegraphics[width=0.9\textwidth]{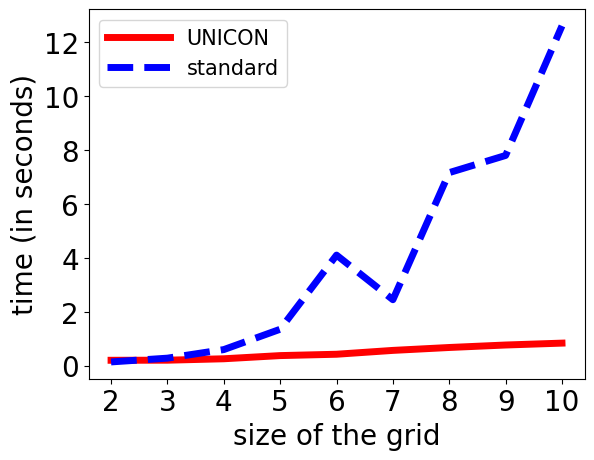}
        \caption{Specification Type 3}
        \label{fig:spec3}
    \end{subfigure}
    \hfill
    \begin{subfigure}[b]{.48\textwidth}
        \includegraphics[width=0.9\textwidth]{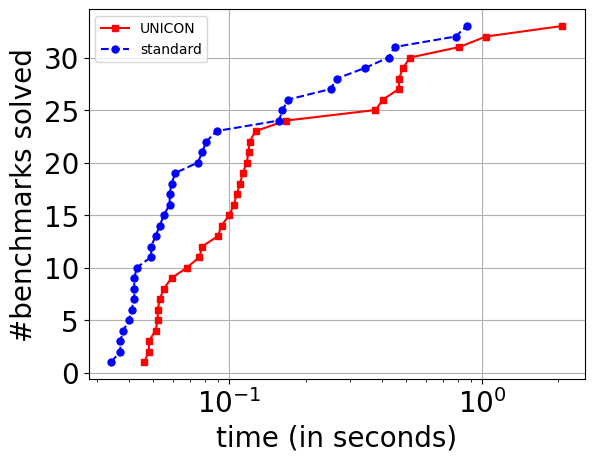}
        \caption{SYNTCOMP benchmarks}
        \label{fig:syntcomp}
    \end{subfigure}
    \caption{We compare the execution times of \prototype{} and the standard synthesis approach. \Cref{fig:spec1} to \Cref{fig:spec3} show the results of three different specification types on increasing robot benchmarks. \Cref{fig:syntcomp} shows a cactus plot for the number of SYNTCOMP benchmarks solved in a certain time budget.}
    \label{fig:maze}
    \vspace*{-0.5cm}
\end{figure}
\section{Experiments}\label{sec:experiments}

We have implemented our approach in a prototype tool called \prototype{}~\cite{tool}. For benchmarking purposes, \prototype additionally implements the standard controller synthesis algorithm that reduces the synthesis problem to game solving~\cite{DBLP:series/natosec/Finkbeiner16}. The tool is implemented in F\# and uses \textsc{FsOmegaLib}~\cite{FsOmegaLib} and \textsc{spot}~\cite{DBLP:conf/cav/Duret-LutzRCRAS22} for automata operations and \textsc{oink}~\cite{DBLP:conf/tacas/Dijk18} for game solving. The experiments are performed on a computer with an Apple M1 Pro 8-core CPU and 16GB of RAM.

We evaluate our approach on a set of benchmarks where the controller guides a robot through a maze-like workspace, avoiding collisions with the walls and with a second (uncontrollable) robot.
Given the size of the grid as a parameter, the workspace is encoded in the plant model as a grid of cells with some walls between the cells and walls around the boundary of the grid.
Initially, the controller robot is placed in the center of the grid and the plant robot in the bottom left corner. At each step, the robots can either stay or move in four directions: up, down, left, and right.
The controller robot has sensors to detect the obstacles (i.e., walls or the plant robot) if they are adjacent to it or when there is a collision. 
The behavioral objective of the controlled robot is specified as an LTL formula. We consider three different types of specifications: Specification type 1 is an invariant $\globally \neg col$, which simply states that the controller must always avoid collisions. In specification type 2, we add some more information about the setting into the temporal formula: $ \neg col \weakuntil \neg \texttt{assump}_\plant$, where the assumption $\texttt{assump}_\plant$ expresses that 
a collision can occur only if in the previous step there was an obstacle in some direction and the controller moved in that direction or stayed in the same position (which might lead to a collision if the obstacle was the other robot).
Specification type 3 is a condition that only restricts the first three steps in the execution, requiring that during those three steps, the robot does not see an obstacle in the \emph{up} direction: $\LTLnext \neg u_{obs} \wedge \LTLnext\LTLnext \neg u_{obs} \wedge \LTLnext\LTLnext\LTLnext \neg u_{obs}$.

\smallskip
\noindent\textbf{Evaluation on robot benchmarks.}
 The results of the experiments are shown in~\Cref{fig:maze}, which shows the comparison of the execution times of \prototype{} (i.e., the total time for synthesizing the universal controller and extracting a correct controller using \cref{alg:general:controller,alg:composition}) with the standard controller synthesis approach to compute a correct controller.
 In all three types of specifications, our approach scales well with growing grid size, while the running time of the standard approach quickly becomes prohibitively large. 
 The reason is that standard controller synthesis constructs the complete composition of the plant and the specification automaton to solve the synthesis problem, while the universal controller synthesis composes the plant and the prophecies on-the-fly. Since the prophecy size is independent of the plant size, the universal controller synthesis scales better with the increasing grid size.
 Specification type 2 performs better than specification type 1, as the added information of the assumption allows for simpler prophecies. Simple specifications like type 3 are most beneficial as the prophecies only need to consider three steps into the future.

 \smallskip
 \noindent\textbf{Baseline evaluation on SYNTCOMP benchmarks.} The benchmarks of the annual SYNTCOMP reactive synthesis competition~\cite{DBLP:journals/corr/abs-2206-00251} are often used to evaluate synthesis algorithms. Since the benchmarks do not contain plant models, we cannot use them to evaluate the scalability of our approach, but we still include them here as a baseline to evaluate the overhead caused by the computation of the universal controller. We selected the 68 benchmarks that have safety requirements and explicit environment assumptions, and translated the assumption formulas into plant models. As shown in~\Cref{fig:syntcomp}, \prototype{} shows similar performance to standard synthesis. Because of the small size of the plant models resulting from the assumption formulas, the overhead for the universal controller is not compensated by the improved performance in the size of the plant model. Still, computation times remain close, showing that \prototype{} is able to scale to practical problem instances. 

\section{Conclusions}

Universal controllers have three major advantages over standard control strategies:
The first advantage is the improved scalability due to the fact that the universal controller is
constructed \emph{independently of the plant}. There is no free lunch -- adapting the universal
controller to a specific plant does require an analysis of the plant
model. However, we only need to verify the prophecy, not the original specification. Simple prophecies can be checked without a full
exploration of the state space. There is no loss of
completeness: the adaptation of the universal controller to a
particular plant succeeds whenever the synthesis problem is
realizable for the plant.

A second advantage is that the universal controller is
\emph{adaptable} to new plants; if the plant model changes, it is not
necessary to re-synthesize the universal controller, only the prophecies need to be re-evaluated.  In principle,
the controller can even be adapted to dynamic changes in the plant
by evaluating the prophecies on-the-fly.

A third major advantage of the universal controller is its
\emph{explainability}: the prophecy identifies the specific property of the
plant that makes it safe to choose the move.  If the adaptation fails,
the violated prophecies identify properties of the plant that 
make it unrealizable.
In future work, we plan to exploit these advantages by constructing controllers that adapt to new situations dynamically and explain their behavior to the user in terms of these changes.

\subsubsection*{Acknowledgements.}
This work is funded by the DFG grant 389792660 as part of TRR 248 – CPEC, by the European Union with ERC
Grant HYPER (No. 101055412), and by the Emmy Noether Grant SCHM 3541/1-1.
Views and opinions expressed are however those of the author(s) only and do not necessarily reflect those of the European Union or the European Research Council Executive Agency. Neither the European Union nor the granting authority can be held responsible for them.

\bibliographystyle{splncs04}
\bibliography{bib}

\end{document}